\documentclass{article} 
\pdfoutput=1
\usepackage{url}

\usepackage{fullpage}
\usepackage{graphicx}
\usepackage{subfigure}

\usepackage{algorithm}
\usepackage[noend]{algorithmic}

\usepackage{amsmath}
\usepackage{amssymb}
\usepackage{amsthm}

\usepackage{multirow}
\usepackage{verbatim}

\usepackage{tikz}
\usetikzlibrary{fadings}
\usetikzlibrary{shapes,snakes}
\usepackage{enumerate}
\usepackage{natbib}
\usepackage{setspace}

\usepackage{hyperref}

\newsavebox\CBox
\def\textBF#1{\sbox\CBox{#1}\resizebox{\wd\CBox}{\ht\CBox}{\textbf{#1}}}

\newtheorem{theorem}{Theorem}

\newtheorem{cor}{Corollary}
\newtheorem{op}{Optimization Problem}

\def\textBF#1{\sbox\CBox{#1}\resizebox{\wd\CBox}{\ht\CBox}{\textbf{#1}}}

\newcommand{\model}{{SPLAM}}
\newcommand{\spam}{{SpAM}}

\title{Sparse Partially Linear Additive Models}

\author{
Yin Lou$^1$ \quad Jacob Bien$^1$ \quad Rich Caruana$^2$ \quad Johannes Gehrke$^1$\\
$^1$Cornell University \quad $^2$Microsoft Research\\
\texttt{\{yinlou, johannes\}@cs.cornell.edu}\\
\texttt{jbien@cornell.edu} \quad \texttt{rcaruana@microsoft.com}\\
}

\def\N{\mathcal{N}}
\def\L{\mathcal{L}}

\begin{document}

\maketitle
\bibliographystyle{natbib}

\bigskip
\begin{abstract}
The generalized partially linear additive model (GPLAM) is a flexible
and interpretable approach to building predictive models.  It combines
features in an additive manner, allowing each to have either a
linear or nonlinear effect on the response.  However, the choice of which features to treat as linear or nonlinear is typically assumed
known.  Thus, to make a GPLAM a viable approach in situations in which
little is known {\em a priori} about the features, one must overcome
two primary model selection challenges: deciding which features to
include in the model and determining which of these features to treat
nonlinearly.  We introduce the sparse partially linear additive model
(\model), which combines model fitting and {\em both} of these model
selection challenges into a single convex optimization problem.
\model~provides a bridge between the lasso and sparse additive models.  Through a statistical oracle inequality
and thorough simulation, we demonstrate that \model~can outperform
other methods across a broad spectrum of statistical regimes, including the
high-dimensional ($p\gg N$) setting.  We develop efficient algorithms
that are applied to real data sets with half a million samples and
over 45,000 features with excellent predictive performance.
\end{abstract}

\section{Introduction}
\label{sec:intro}

Generalized partially linear additive models (GPLAMs, \citealt{hardle2007plm}) provide an
attractive middle ground between the simplicity of generalized linear
models (GLMs, \citealt{nelder1972generalized}) and the flexibility of generalized additive models
(GAMs, \citealt{hastie1990gam}).  Given a data set $\{(x_i, y_i)\}_{i = 1}^N$, a GPLAM relates the conditional mean of the response, $y_i$, to the $p$-dimensional predictor vector, $x_i$, using a known link function, $g$:
\begin{align}
  g(E[y_i|x_i]) = \sum_{j\in\N}f_j(x_{ij}) + \sum_{j\in\L}x_{ij}\theta_j.\label{eq:gplam}
\end{align}
The features in $\N$ contribute to the model in a nonlinear fashion while the features in $\L$ contribute in a linear fashion.  A GLM treats all features as being in $\L$ and may therefore be biased when nonlinear effects are present; on the other extreme, a GAM treats all features as being in $\N$, which incurs unnecessary variance for the features that should be treated as linear.  GPLAMs are a popular tool for data analysis in multiple domains including economics~\citep{engle1986semiparametric, green1993nonparametric} and biology~\citep{lian2012identification, dinse1983regression}.

A major obstacle to using GPLAMs on large-scale data sets is that one rarely
knows {\em a priori} which features should be assigned to $\N$ and $\L$.  
A further challenge is in deciding which features should be excluded from the model entirely.
The goal of this paper is to make GPLAMs a viable tool for building large-scale predictive models.  To do so, we must overcome two model-selection
challenges: automatically deciding which features are at all relevant in the model and deciding which of those features should be fit
linearly versus nonlinearly.

In the context of GAMs (where $\L$ is taken to be empty), the sparse additive model (\spam) is a useful framework for performing feature selection on $\N$ \citep{ravikumar2009spam}.  From the perspective of a GPLAM, \spam~takes
an ``all-in" or ``all-out" approach to feature selection.  In this work, we introduce the sparse partially linear additive model (\model) that provides the finer-grained selection demanded by a GPLAM.  \model s~build on the \spam~framework, providing a natural bridge between the $\ell_1$-penalized GLM and \spam, thereby reaping many of the benefits enjoyed by both of these methods.

Failing to account for exactly linear features is disadvantageous statistically, computationally, and in terms of interpretability. As a
motivating example, consider a situation in which $p=1,000$, $|\N|=5$,
and $|\L|=295$.  Assuming the correct set of features is selected,
\spam~would include 300 features.  From an interpretability
standpoint, one would have to manually inspect the 300 nonparametric
fits to reveal that only 5 features are effectively nonlinear.  The
other 295 of them would appear nearly, but not exactly, linear such as
in Figure~\ref{fig:synth} (d).  Statistically, a price is paid in variance for the many nearly-linear features; and, computationally, such a model is wasteful both in terms of memory and speed for making future predictions.

In the last several years, a number of methods have been proposed to address various aspects of this problem. 
   In \citet{chen2011determination}, a bootstrap-based test is developed to determine the linearity of a component.  In \citet{huang2012semiparametric}, the authors use a group MCP penalty to decide which features should be linear versus nonlinear, but features may not be completely excluded from the model.   
In \citet{du2012semiparametric}, an algorithm is developed that iterates between two optimization problems: one that decides which nonlinear features should be made linear and the other that decides which linear features should be set to zero.  
An alternative approach to \spam~is the component selection and smoothing operator (COSSO) method, which uses unsquared reproducing kernel Hilbert space (RKHS) norm penalties \citep{lin2006component}.  The linear and nonlinear discoverer extends COSSO to the GPLAM setting \citep{zhang2011linear}.  Relatedly, \citet{lian2012identification} combine smoothness and sparsity SCAD-based penalties for a similar purpose.
None of the above methods is geared toward high-dimensional data in terms of statistical theory or computation.  
Several other methods are geared toward the high-dimesional setting but do not perform both model selection tasks.  For example, in \citet{bunea2004consistent,xie2009scad,muller2013partial}, methods are developed to
perform feature selection for the set $\L$ while assuming that the set $\N$
is known; \citet{lian2013generalized} and~\citet{wang2014estimation} perform feature selection on
both $\L$ and $\N$ individually but assume an initial partition of the
features into those potentially in $\L$ and those potentially in $\N$.

By contrast, \model~is designed for large-scale datasets (for example, we apply it to a dataset with $p=47,236$).  \model~is formulated through a single convex optimization problem that admits an efficient algorithm and strong theoretical properties even in the $p\gg N$ setting.

In Section \ref{sec:optimization}, we define the \model~estimator as the solution to a convex optimization problem, and, in Section \ref{sec:computation}, we discuss how this problem may be efficiently solved in large-scale contexts.  Section \ref{sec:theory} presents consistency results under weak assumptions and lends theoretical support to the conceptual difference in predictive performance between \model~and \spam, its close relative.  Section \ref{sec:experiments} provides an empirical study of \model, including both a thorough simulation study and comparison of \model~to other methods on an array of large data examples.

\section{The \model~Optimization}
\label{sec:optimization}

We approach the challenging model selection and fitting problem posed by a GPLAM through convex relaxation.  For each feature $x_j$, we perform an $M$-dimensional basis expansion $b(x_j) = [b_1(x_j), ..., b_M(x_j)]$ in which $b_1(x_j) = x_j$ and $M$ is typically small. Our main requirement of the basis is that $b_1(x_j)$ models the linear part and that $[b_2(x_j), ..., b_M(x_j)]=:b_{-1}(x_j)$ models the nonlinear part. 
\model~estimates each $f_j(\cdot)$ by a function in the space spanned by $b(\cdot)$, i.e., $f_j(x_j) = b(x_j)\beta_j$, where $\beta_j\in\mathbb{R}^M$ (for ease of exposition we ignore the intercept).  We use $\beta_{j1}$ and $\beta_{j,-1}$ to denote the coefficients of the linear and nonlinear basis functions, respectively. Letting $\beta = [\beta_1^T, ..., \beta_p^T]^T\in\mathbb{R}^{pM}$ and $X=[b(x_1):\cdots:b(x_p)]\in\mathbb{R}^{N \times pM}$ be the design matrix, we have 
$$
X\beta=\sum_{j=1}^p\left[\beta_{j1}x_j+\beta_{j,-1}b_{-1}(x_j)\right].
$$
Given a convex smooth loss function $L(y, X, \beta)$, \model~is formulated as the solution to the following convex program with hierarchical sparsity regularization:
\begin{op}{\model}\label{op:splam}
\begin{align}
\min_{\beta} L(y, X, \beta) + \lambda \Omega^{\model}(\beta)
\end{align}
\end{op}
where 
$
\Omega^{\model}(\beta) = \sum_{j=1}^p\left[\alpha\|\beta_j\|_2 + (1 - \alpha)\|\beta_{j, -1}\|_2\right],
$
 $\lambda \ge 0$, and $\alpha\in[0,1]$.

In this paper, we focus on linear regression, in which $L(y, X, \beta) = \frac{1}{2N}\|y - X\beta\|_2^2$ and logistic regression, in which $L(y, X, \beta) = \frac{1}{N}\sum_i \log(1 + \exp[-y_i\sum_{k=1}^{pM}X_{ik}\beta_k])$, where $y_i \in \{-1, 1\}$. The penalty function $\Omega^{SPLAM}$ is convex and is an instance of the hierarchical group lasso \citep{zhao2009composite,jenatton2010proximal}.  Its two terms address the two forms of model selection present in the GPLAM problem: the first term affects the overall number of predictors appearing in the fitted model; the second term controls the number of those features that are treated nonlinearly. 

Just as GPLAMs generalize both GLMs and GAMs, it is useful to note that \model~includes the most common penalized estimators used for these two kinds of models.
\begin{itemize}
\item When $\alpha = 1$ and an orthogonal basis is used, Problem~\ref{op:splam} becomes \spam~in group lasso form.
\item When $\lambda=\tilde\lambda/\alpha$ and $\alpha$ is sufficiently small, \model~reduces to the lasso \citep{tibshirani1996regression} applied to the linear features only.
\end{itemize}

In practice, we solve the \model~problem over a grid of $(\lambda, \alpha)$ pairs. Our strategy is to fix $\alpha$ and solve the problem pathwise starting from the smallest value of $\lambda$ for which $\hat\beta_j = 0$ for all $j = 1, ..., p$ and decreasing $\lambda$ exponentially.  In Section~\ref{sec:oracle-inequality}, we prove that \model~is consistent under general conditions for $\alpha=(1+\sqrt{6})/(1+2\sqrt{6})$ and suitably chosen $\lambda$.

After posting the initial draft of this paper online, we learned of a
similar method being developed independently and concurrently to ours
\citep{chouldechova2015generalized}; their approach to the GPLAM problem also
makes use of an overlapping group lasso penalty, but uses a different
form of penalty known as the {\em latent overlapping group lasso
  penalty} \citep{Obozinski2011}. Also, \citet{Petersen2014flam} in a recent preprint propose a method that combines feature selection and non-linear, piecewise-constant modeling using a fused-lasso penalty.

\section{Computation}
\label{sec:computation}

The hierarchical group lasso can be solved efficiently by proximal gradient descent~\citep{beck2009fista} as described in~\citet{jenatton2010proximal}.
The idea of this algorithm is to modify the standard gradient steps that one would take if simply minimizing $L$ and then apply the proximal operator of the nondifferentiable penalty, $\lambda\Omega^{\model}(\cdot)$:
\begin{align}
\beta^{k + 1} = \arg\min_{z\in\mathbb{R}^{pM}} \{\frac{1}{2t^k}\|z - (\beta^k - t^k\nabla L(y, X, \beta^k))\|_2^2 + \lambda\Omega^{\model}(z)\},
\end{align}
where $t^k$ is a suitable step size. It is known that setting the step
size to the reciprocal of the Lipschitz constant of $\nabla L$
guarantees convergence~\citep{beck2009fista}. A key property of
hierarchical penalties such as $\Omega^{\model}$ is that the proximal
operator can be very efficiently solved.   In particular,
\citet{jenatton2010proximal} show that the dual of this problem can be
solved in a single pass of block coordinate descent (and therefore has
essentially a closed form).  While the proximal gradient method as
described above can be used to solve this problem, we observe that a
closely related method, called {\em block coordinate gradient descent}
performs better in practice for solving the \model~problem in
large-scale settings.  Furthermore, in the regression setting,
we develop an even more efficient approach that solves the problem by
applying the proximal operator only once.  
Additional details about our implementation are given in the supplementary material.

\subsection{Block Coordinate Gradient Descent}
\label{sec:bcgd}

The block coordinate gradient descent (BCGD) method is a hybrid of blockwise coordinate descent (BCD) and a proximal method.  A simple quadratic approximation of $L$ is used in each coordinate update.  The particular form of BCGD we propose is to apply the proximal operator one block at a time, allowing each block update to use a distinct step size.  We find that empirically this is more efficient than proximal gradient descent (this has been noted in a related problem by \citealt{qin2010efficient}, in which they call this method ISTA-BC).

We cycle through the blocks (taking each $\beta_j\in\mathbb{R}^M$ as a
block), and on the $(k+1)$st pass, the update of block $j$ is given by
\begin{align}
\beta_j^{k + 1} = P_{t_{j}}^j(\beta^{k}) =: \arg\min_{z\in\mathbb{R}^{M}} \left\{\frac{1}{2t_j}\|z - (\beta_j^k - t_j\nabla_{\beta_j} L(\beta^k))\|_2^2 + \lambda\alpha\|z\|_2 + \lambda(1 - \alpha)\|z_{-1}\|_2\right\}
\end{align}
where $t_{j}$ is the step size for block $j$.

This proximal problem has essentially a closed-form solution and therefore can be solved very efficiently, as shown in \citet{jenatton2010proximal}.  Let $g_j = \beta_j^{k} - t_{j}\nabla_{\beta_j} L(\beta^k)$ and consider its dual,
\begin{align}
\min_{\gamma_1\in\mathbb{R}^{M}, \gamma_2\in\mathbb{R}^{M-1}}~&~\frac{1}{2}\|g_j - \gamma_1 - [0, \gamma_2^T]^T\|_2^2\\
s.t.~&~\|\gamma_1\|_2 \le t_{j}\lambda\alpha \qquad \|\gamma_2\|_2 \le t_{j}\lambda(1 - \alpha).
\end{align}

\citet{jenatton2010proximal} show that it can be solved in \emph{one pass} of block coordinate descent,
\begin{align}\label{eqn:dual_solution}
\hat\gamma_2~=~\Pi_{t_{j}\lambda(1 - \alpha)} (g_{j,-1}), \hat\gamma_1~=~\Pi_{t_{j}\lambda\alpha} (g_j - [0, \hat\gamma_2^T]^T)
\end{align}
where $\Pi_r(u)$ is the Euclidean projection of the vector $u$ onto
the $\ell_2$-ball of radius $r$. Having solved the dual, we get $P_t^j(\beta^k) = g_j - \hat\gamma_1 - [0, \hat\gamma_2^T]^T$.

We perform a backtracking line search until the following inequality holds to select $t_j$:
\begin{align}\label{eqn:backtracking}
L(\tilde{\beta}) \le L(\hat{\beta}) + \langle P_{t_{j}}^j(\beta^{k}) - \beta_j^{k}, \nabla_{\beta_j} L(\hat{\beta})\rangle + \frac{1}{2t_{j}}\|P_{t_{j}}^j(\beta^{k}) - \beta_j^{k}\|_2^2,
\end{align}
where
\begin{align*}
\hat{\beta}~=~[{\beta_1^{k + 1}}^T, ..., {\beta_{j-1}^{k + 1}}^T, {\beta_j^{k}}^T, ..., {\beta_p^{k}}^T]^T\text{~and~}\tilde{\beta}~=~[{\beta_1^{k + 1}}^T, ..., {\beta_{j-1}^{k + 1}}^T, P_{t_{j}}^j(\beta^{k})^T, ..., {\beta_p^{k}}^T]^T.
\end{align*}

Computing the Lipschitz constant, $C_j$, of $\nabla_{\beta_j} L$ is
relatively inexpensive since $X_{j}^TX_{j}$ is just an $M$-by-$M$
matrix, where $M$ is typically very small.  Thus, in practice we can
easily compute the minimum step size $1 / C_j$, to avoid the step size
$t_{j}$ going below this value.

Algorithm~\ref{algo:bcg} summarizes our BCGD method. We cycle through each block (Line 5), solve the proximal operator for that block (Line 7-10) and check if the step size is proper using a backtracking line search (Line 11-14). In the supplementary material, we show that the proposed algorithm
fits the framework of \citet{tseng2009coordinate} and therefore is guaranteed to converge.

\begin{algorithm}[t]
\caption{\em\small \model~via BCGD (for general loss $L$)}
\label{algo:bcg}
\begin{algorithmic}[1]
\STATE $t_j = t_j^0$, for $j = 1, ..., p$
\STATE $k \leftarrow 0$
\STATE $\beta^0 \leftarrow 0$
\WHILE{not converge}
    \FOR{$j = 1$ to $p$}
        \WHILE{\textbf{true}}
            \STATE $g_j \leftarrow \beta_j^{k} - t_{j}\nabla_{\beta_j} L(\beta^{k})$
            \STATE $\hat\gamma_2 \leftarrow \Pi_{t_{j}\lambda(1 - \alpha)} (g_{j,-1})$
            \STATE $\hat\gamma_1 \leftarrow \Pi_{t_{j}\lambda\alpha} (g_j - [0, \hat\gamma_2^T]^T)$
            \STATE $\beta_j^{k + 1} \leftarrow g_j - \hat\gamma_1 - [0, \hat\gamma_2^T]^T$
            \IF{Inequality~\eqref{eqn:backtracking} holds}
                \STATE \textbf{break}
            \ELSE
                \STATE $t_j \leftarrow \min(\eta t_j, 1 / C_j)$
            \ENDIF
        \ENDWHILE
    \ENDFOR
    \STATE $k \leftarrow k + 1$
\ENDWHILE
\end{algorithmic}
\end{algorithm}

\subsection{Block Coordinate Descent}
\label{sec:bcd}

Although Algorithm~\ref{algo:bcg} is applicable to any differentiable
loss function $L$, in the special case of a quadratic loss, a more
efficient solution strategy is available if we are willing to use an
orthonormal basis expansion, $Q_{j}\in\mathbb{R}^{n\times M}$, of each
feature $j$.  Thus, in this section we assume that the design
matrix $X=[Q_1:\cdots:Q_p]$ and that $Q_{j}^TQ_{j} = I_M$.  (We still
require, as throughout this paper, that the first column corresponds to the linear term.)  In block coordinate descent, we cycle through the $\beta_j$'s and for the $j$th block, solve the subproblem
\begin{align}
\min_{\beta_j} \frac{1}{2N} \|r_j - Q_{j}\beta_j\|_2^2 + \lambda\alpha\|\beta_j\|_2 + \lambda(1 - \alpha)\|\beta_{j, -1}\|_2
\end{align}
where $r_j = y - \sum_{k \neq j} Q_{k}\beta_k$ is the $j$th partial residual.

For general $Q_j$, this update would require an iterative approach,
but since $Q_{j}$ is an orthogonal matrix, we can equivalently minimize
\begin{align}
\min_{\beta_j} \frac{1}{2N} \|Q_{j}^Tr_j - \beta_j\|_2^2 + \lambda\alpha\|\beta_j\|_2 + \lambda(1 - \alpha)\|\beta_{j, -1}\|_2,
\end{align}
which we recognize as the optimization problem from BCGD, in which we apply the proximal operator to $Q_{j}^Tr_j$ instead of to $\beta_j - t_j\nabla_{\beta_j}L(\beta)$. Thus, by using BCD instead of BCGD we obviate the need to select a step size, making the optimization much more efficient.

To get the orthonormal basis, $Q_j$, we begin with a basis $X_j$ and
then perform a QR decomposition for
each block $j$ using the Gram-Schmidt process in order to preserve the
linear basis in the first column of each block.  Algorithm~\ref{algo:bcd} summarizes our BCD algorithm. 

\begin{algorithm}[t]
\caption{\em\small \model~via BCD (for quadratic loss $L$ and $Q_j\in\mathbb{R}^{n\times M}$ orthonormal)}
\label{algo:bcd}
\begin{algorithmic}[1]
\STATE $\beta^0 \leftarrow 0$
\WHILE{not converge}
    \FOR{$j = 1$ to $p$}
       \STATE $g_j \leftarrow Q_{j}^Tr_j$, where $r_j = y - \sum_{k \neq j} Q_{k}\beta_k$
       \STATE $\hat\gamma_2 \leftarrow \Pi_{N\lambda(1 - \alpha)} (g_{j,-1})$
       \STATE $\hat\gamma_1 \leftarrow \Pi_{N\lambda\alpha} (g_j - [0, \hat\gamma_2^T]^T)$
        \STATE $\beta_j \leftarrow g_j - \hat\gamma_1 - [0, \hat\gamma_2^T]^T$
    \ENDFOR
\ENDWHILE
\end{algorithmic}
\end{algorithm}

\section{Statistical Theory}
\label{sec:theory}

In this section, we seek a deeper understanding of the regimes in which \model~works well.  In Section \ref{sec:oracle-inequality}, we prove an upper bound on \model's prediction error in the regression setting.  This establishes \model~as a reliable method even when $p\gg N$ and gives insight into the factors that influence its prediction performance.  In Section \ref{sec:impr-over-spam}, we consider an asymptotic regime that highlights \model's potential statistical advantage over \spam.

\subsection{Oracle Inequality}
\label{sec:oracle-inequality}
Oracle inequalities have been proved for the hierarchical group lasso
(see, e.g., \citealt{chatterjee2012sparse}) that could be applied to
\model.  These results follow from the unified framework
of~\citet{negahban2012unified}, which gives both oracle inequalities
and recovery guarantees for a wide class of estimators based on
decomposable regularizers.  However, such results (and others of its
kind) make potentially strong (and unverifiable) assumptions on the
design matrix (e.g., the restricted isometry
property~\citealt{candes2007dantzig}, the compatibility
condition~\citealt{buhlmann2011statistics, van2007deterministic},
small coherence~\citealt{candes2009near}, etc.  See
\citealt{van2009conditions}).  Since \model's design matrix consists
of derived features, such assumptions become even more difficult to
interpret.  There is, however, a different class of oracle
inequalities, known as ``slow rates'', that make {\em no assumptions}
on the design matrix~\citep{dalalyan2014prediction}.  In addition, despite their name, these inequalities have been shown in some cases to give faster rates of convergence than the more standard ``fast rates'' \citep{dalalyan2014prediction}.  They are particularly useful 
in situations where the various assumptions made by the fast rate bounds are known not to apply or would be particularly difficult to interpret.

We derive in this section slow rate bounds for \model, thereby giving us an understanding of its statistical performance under no conditions on $X$. To the best of our knowledge, these are the first such slow rate bounds derived for the hierarchical group lasso.

Suppose
$$
y_i = f^0(x_i)+\epsilon_i\text{ for }i=1,\ldots, N,
$$
where $x_i\in\mathbb R^p$ is a vector of features, $\epsilon_i \sim N(0,\sigma^2 I_N)$ is a random
 vector of noise, and $f^0$ is the underlying function.  
Let $\hat\beta\in\arg\min_\beta\left\{\frac1{2N}\|y-\sum_jX_j\beta_j\|_2^2+\lambda\Omega^{\model}(\beta)\right\}$
denote a solution of \model~in which we have
 orthogonalized each feature's design matrix, i.e., that
 $\frac1{N}X_j^{T}X_j=I_M$ and let $\hat
 f=\sum_jX_j\hat\beta_j\in\mathbb R^N$
 denote the set of fitted values at these $N$ points.  The following
 theorem provides a slow rate for \model's prediction error. In an abuse of
 notation, we let $f^0$ denote the vector with $i$th element given by $f^0(x_i)$.

\begin{theorem}\label{thm:oracle_slow_rate}
If we take $\lambda\ge 2(1 + 2\sqrt{6})\sigma\sqrt{\log p / N}$ and
$\alpha=(1 + \sqrt{6}) / (1 + 2\sqrt{6})$, then
\begin{align}
\frac1{N}\|\hat f-f^0\|^2\le \min_{\beta\in\mathbb R^{p M}}\left\{\frac1{N}\| f^0-\sum_{j=1}^pX_j\beta_j\|_2^2+3\lambda\Omega^{\model}(\beta)\right\}
\end{align}
holds with probability at least $1-4/p$ as long as $\log p \ge M/8$.
\label{thm:oracle}
\end{theorem}
\begin{proof}
  See supplementary material.
\end{proof}
The above theorem makes no assumptions about the underlying function
$f^0$, and shows that \model~works well if there exists $\beta$ for
which $\sum_jX\beta_j$ is not too far from $f^0$ and
$\Omega^{\model}(\beta)$ is small. In the special case
that $f^0=\sum_jX_j\beta_j^0$ for some sparse vectors
$\beta_1^0,\ldots,\beta_p^0$, the result takes a simpler form.
We describe the sparsity of $\beta^0\in\mathbb{R}^{pM}$ in two senses: first, in terms
 of whether a feature is at all relevant,
 $\mathcal{S}_0=\{j:\beta^0_j\neq0\}$, and, second, in terms of
 whether the feature is nonlinear,
 $\mathcal{N}_0=\{j:\beta_{j,-1}^0\neq0\}$.  We also define the set of
 linear features, $\mathcal{L}_0=\mathcal{S}_0\setminus
 \mathcal{N}_0$.  Under this stronger assumption on $f^0$, the
 statement simplifies greatly, revealing the roles that $\mathcal{L}_0$ and
 $\mathcal{N}_0$ play in the performance of the estimator.

\begin{cor}\label{thm:slow_rate}
Suppose $f^0=\sum_jX_j\beta_j^0$ with $\mathcal{L}_0$ and
$\mathcal{N}_0$ defined as above.  If we take $\lambda\ge 2(1 + 2\sqrt{6})\sigma\sqrt{\log p / N}$ and
$\alpha=(1 + \sqrt{6}) / (1 + 2\sqrt{6})$, then
\begin{align}
\frac1{N}\|\hat f - f^0\|_2^2\le 3\lambda\left[\alpha\sum_{j\in \mathcal{L}_0}|\beta^0_{j1}|+\sum_{j\in \mathcal{N}_0}\|\beta^0_{j}\|_2\right]
\end{align}
holds with probability at least $1-4/p$ as long as $\log p \ge M/8$.
\label{thm:oracle}
\end{cor}
\begin{proof}
 See supplementary material.
\end{proof}
The above corollary implies that for suitably chosen $\lambda$, \model's prediction error converges to 0 in probability as
  $N\to\infty$ even if we let $p$ grow like $e^{N^\gamma}$ with $\gamma<1$ (assuming the sets $\L_0$ and $\N_0$ and the coefficients of features in this set remain fixed).  
 It also shows that our error grows linearly in the number of both linear and nonlinear features in the true model.  An interesting implication of the theorem is that $\alpha\approx0.58$ is a theoretically justifiable choice (although better performance may be achievable by tuning $\alpha$).

When all features are linear ($\N_0=\emptyset$), this result reduces to the traditional slow rate bound for the lasso (up to constants) \citep{rigollet2011exponential}.  Such bounds have been improved for the lasso by careful incorporation of the design matrix \citep{hebiri2013correlations}, and we speculate that similar improvement could be developed here.

\subsection{A Comparison to~\spam~When All Features Are Linear}
\label{sec:impr-over-spam}
We have seen in the previous section that \model~is consistent in prediction error in the presence of both linear
and nonlinear features even when $p\gg N$.  Since \spam~is a special
case of \model~(with $\alpha=1$), similar bounds follow easily.  A
natural question then is whether there is any {\em statistical} reason
to prefer \model~over \spam~(aside from the easier interpretation of a
GPLAM over a GAM when many features are linear).  Intuitively, it seems that when many features are truly linear, \spam~incurs variance for estimating nonlinear terms without a useful
reduction in bias; on the other hand, for \model~this would not happen, assuming a sufficiently large parameter for the nonlinear-specific penalty.  We make this intuition more precise by considering a scenario in which \model~is consistent whereas
\spam~is not, implying that there is indeed a statistical advantage to using \model.

Suppose that all $p$ features are linear with equal coefficients,
i.e., $\beta_{j}^0=be_1\in\mathbb R^M$,
and consider an asymptotic regime in which $p$ is fixed and $N=pM$ with $M,N\to\infty$ (note, Theorem~\ref{thm:oracle} does not apply since here $M>8\log p$).  We assume that all
features are orthogonal, i.e. $\frac1{N}X^TX=I_{N}$, so that \spam~
has a simple closed-form expression:
$$
\hat\beta_j^{\spam}=\left(1-\frac{\lambda}{\|\frac1{N}X_j^Ty\|_2}\right)_+\frac1{N}X_j^Ty.
$$
A several-line argument in the supplementary material establishes that
$$
\lim_{N,M\to\infty}\frac1{N}\|X\hat\beta^{\spam}-X\beta^0\|_2^2\ge \frac{b^2}{1/b^2+p/\sigma^2}>0.
$$
Thus, in the asymptotic regime in which one allows the number of basis vectors to grow linearly with $N$, one finds that the prediction error is bounded away from zero (regardless of the choice of $\lambda$).  Interestingly, this lower bound matches (up to constants) the upper bound for the group lasso in Theorem 8.1 of \citet{buhlmann2011statistics}.

By contrast, consider \model~with $\lambda\alpha=0$ and
$\lambda(1-\alpha)=\infty$ (e.g., take $\alpha\to0$ and
$\lambda=\alpha^{-1/2}$).  With this choice of parameters, it is apparent that
$\hat\beta^{\model}_j=X_{j1}^Ty\cdot e_1$ is simply the least squares solution on the correct set of variables and
$$
\frac1{N}\|X\hat\beta^{\model}-X\beta^0\|_2^2\to0.
$$
While assuming that the number of basis functions, $M$, is growing
linearly in $N$ is of course particularly unfavorable to
\spam~(indeed, \citealt{ravikumar2009spam} note that $M$ growing like
$N^{1/5}$ is a standard choice), it does serve to support the intuition
regarding the statistical cost of incorrectly assuming nonlinearity.
Indeed, in Section \ref{sec:synthetic} (Figure \ref{fig:simulation})
we show that there is a wide range of
scenarios in which \model~does in fact have better performance than \spam.

\section{Empirical Study}
\label{sec:experiments}
In this section, we report experimental results for \model. For all our experiments, we use cubic splines with 10 knots for basis expansion: $b(x_j) = [x_j, x_j^2, x_j^3, (x_j - x_{j1}^*)_+^3, ..., (x_j - x_{j10}^*)_+^3]$ (i.e., $M = 13$), where $(\cdot)_+$ represents the non-negative part and  the knot $x_{j\cdot}^*$ is chosen from quantiles in the sample. We choose the best parameters on a held-out validation set and report model performance on a test set. The code is available at \url{https://github.com/yinlou/mltk}.

\subsection{Synthetic Problem}
\label{sec:synthetic}
To illustrate the use of \model, we generate $N=2,000$ points from the model  $y = 2\sin(2x_1) + x_2^2 + \exp(-x_3)  + x_4 - 3x_5 + 2.5x_6 + 10x_7 + 2x_8 - 7x_9 + 5x_{10} + \epsilon$, where $\epsilon\sim\mathcal{N}(0, 1)$. In this experiment, we create an additional 90 random features (so $p=100$). The first 3 nonlinear features are generated uniformly in $[-2.5, 2.5]$, and all other features are uniformly in $[0, 1]$.

\begin{figure*}[t]
\begin{center}
\begin{tabular}{ccccc}
    \includegraphics[height=35mm]{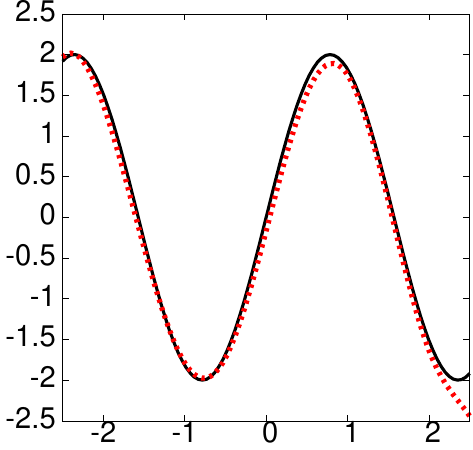} & \includegraphics[height=35mm]{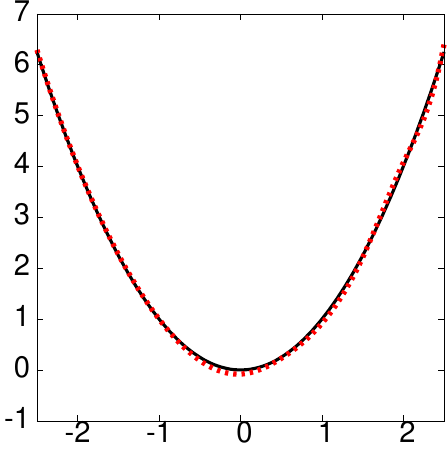} & \includegraphics[height=35mm]{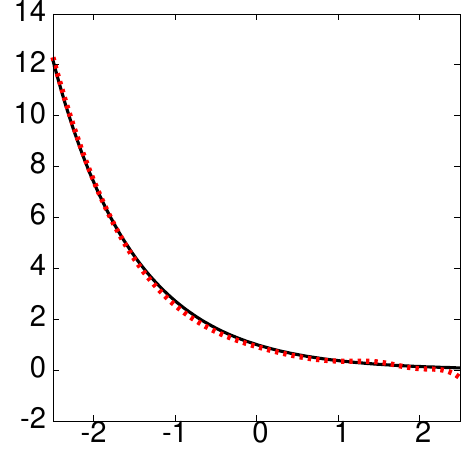} & \includegraphics[height=35mm]{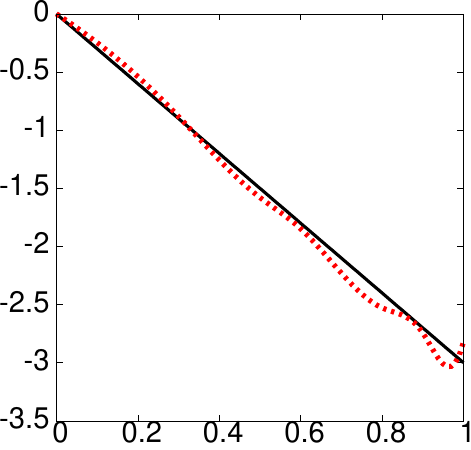}\\
    (a) $f_1$ & (b) $f_2$ & (c) $f_3$ & (d) $f_5$ in \spam
\end{tabular}
\end{center}
\caption{\em\small Estimated component functions (in dashed red) and true functions (in solid black)
 for synthetic dataset in Section~\ref{sec:synthetic}. Nonlinear estimates of \model~are illustrated in (a) - (c). Figure (d) shows \spam's estimate of $f_5$.  The wiggliness of \spam's estimate is because it does not penalize toward exact linearity as does \model.} 
\label{fig:synth}
\end{figure*}

We plot estimated components in Figure~\ref{fig:synth}. Figure~\ref{fig:synth} (a), (b), and
(c) visualize the nonlinear components in \model~for $f_1$, $f_2$, and
$f_3$, respectively. We can see that the estimated shape of the component
function is very close to the true functions. On this sample,
\model~perfectly recovers which features are linear and nonlinear
while \spam~treats all selected features as nonlinear. For
coefficients on linear components in \model, the relative error is
less than 0.1\%. For comparison, we visualize $f_5$ in \spam~in red in
Figure~\ref{fig:synth}(d). The ground truth linear function is plotted
in black. We can see that the component itself is not exactly linear
and that it overfits to the noise. 

\subsection{Simulation: Effect of $|\N|$ and $|\L|$ on~\model,~\spam, and the Lasso}
\label{sec:simulation}
In this section, we perform a large-scale simulation to gain deeper insights into the lasso~\citep{tibshirani1996regression}, \model~and \spam~\citep{ravikumar2009spam}. We consider the models with $p = 100$ features: $y = \sum_{j\in\mathcal{L}} x_j + \sum_{j\in\mathcal{N}} \sin(x_j) + \epsilon$, where $\epsilon\sim\mathcal{N}(0, 1)$, $\mathcal{L}\cap\mathcal{N} = \emptyset$. We use two parameters $\gamma$ and $\delta$ to control the cardinality of $\mathcal{L}$ and $\mathcal{N}$, respectively, i.e., $|\mathcal{L}| = \gamma p$ and $|\mathcal{N}| = \delta p$. We choose $\gamma = 0.0, 0.1, ..., 1.0$ and $\delta = 0.0, 0.1, ..., 1.0$ ($\gamma + \delta \le 1$) and for each $(\gamma, \delta)$ pair, we generate 10 different models. For each of those models, we generate $N_{train}$ points for training, $N_{valid}$ points for validation and $N_{test}$ points for testing. We consider three different settings of simulations, $(N_{train}, N_{valid}, N_{test}) = (200, 100, 100)$, $(500, 100, 100), (1000, $ $200, 200)$. 

For \model, we consider $\alpha \in \{0.05$, $0.1, ..., 0.95, 1\}$. For all of the three methods, we consider the full regularization path with 100 $\lambda$s spaced evenly on a log scale. In our experiments, this range is sufficient to find the optimal model structure. Best model parameters are chosen using the validation set, and model accuracy is evaluated as the average RMSE of 10 models on test sets. For all successive experiments, we use the same method to choose parameters.

\begin{figure*}[t]
\begin{center}
\begin{tabular}{cccc}
  & $N_{train} = 200$ & $N_{train} = 500$ & $N_{train} = 1000$ \\
\rotatebox{90}{~~~~~~Nonlinear Terms (\%)} & \includegraphics[width=0.28\textwidth]{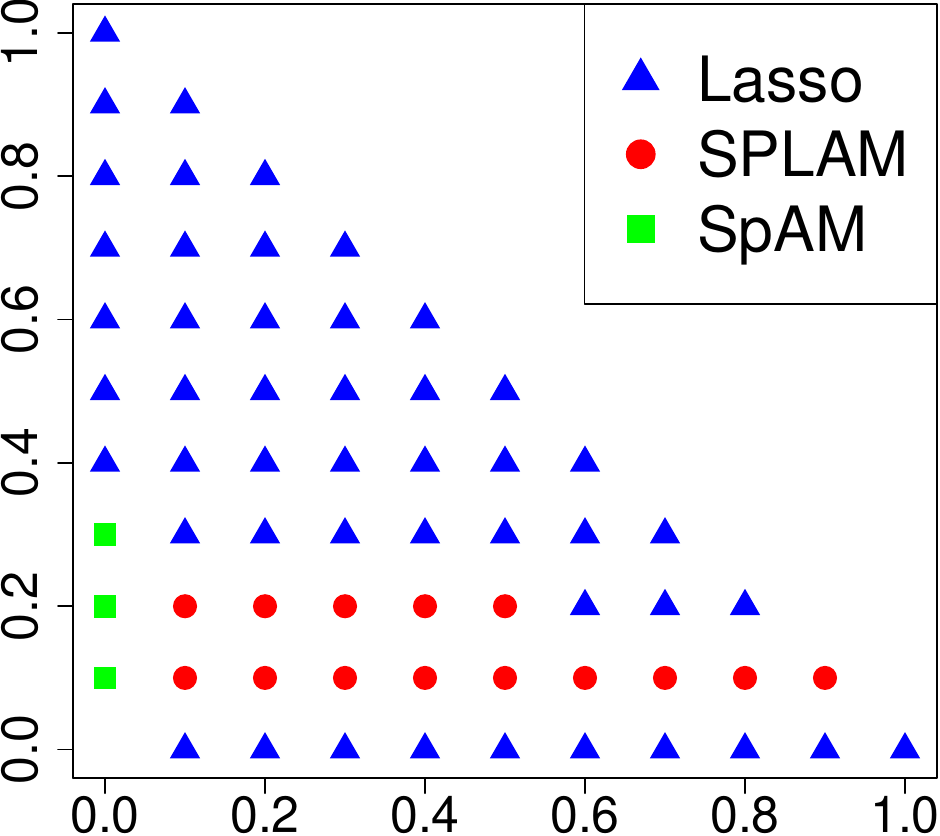} & \includegraphics[width=0.28\textwidth]{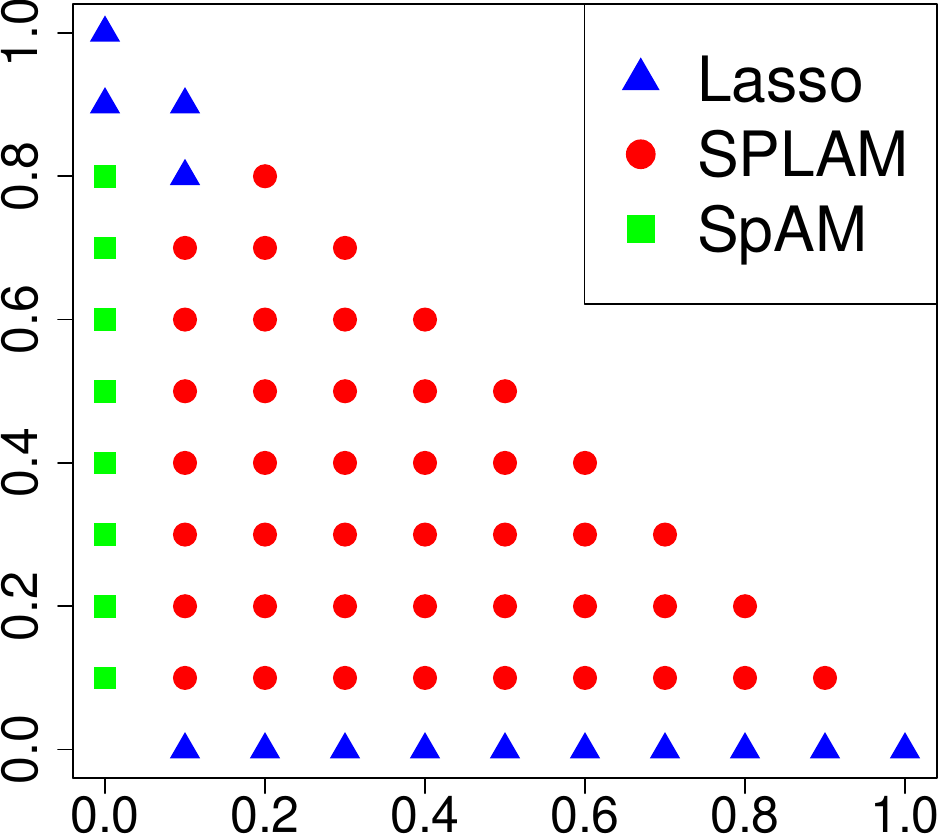} & \includegraphics[width=0.28\textwidth]{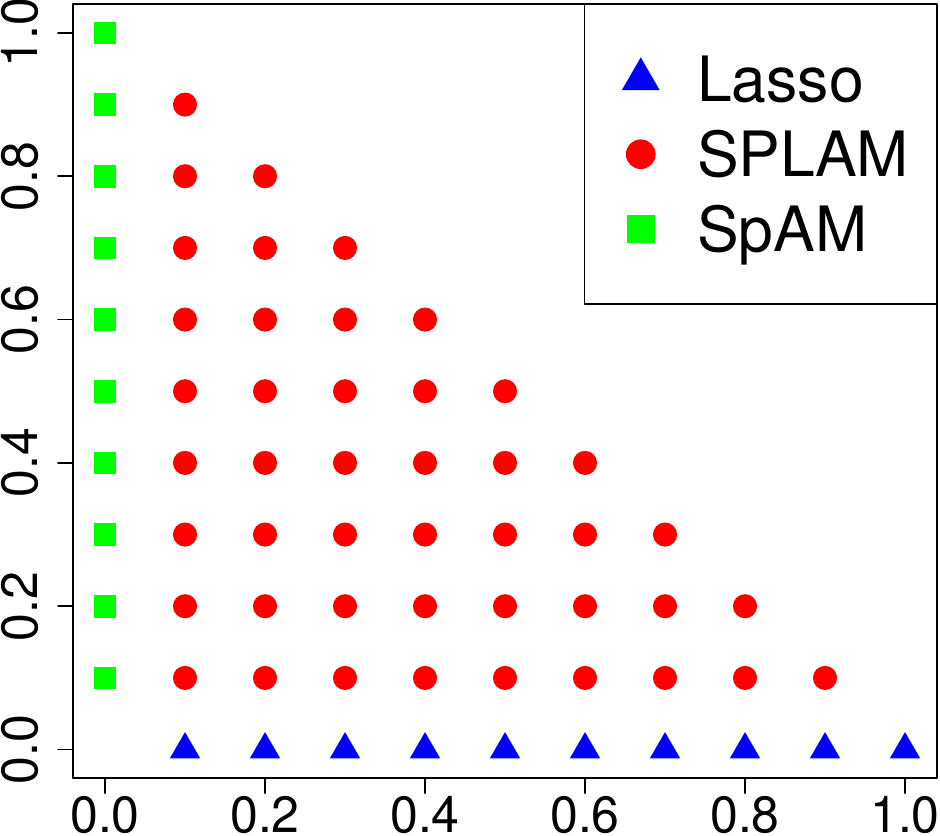} \\
& \multicolumn{3}{c}{Linear Terms (\%)}  \\
\end{tabular}
\caption{\em\small Results of simulation in Section~\ref{sec:simulation}. Each point shows the winning model for a given $(\gamma, \delta)$ pair.}
\label{fig:simulation}
\end{center}
\end{figure*}

Figure~\ref{fig:simulation} shows the results for the simulations. For each $(\gamma, \delta)$ pair, we plot which model wins on average. It is clear that for pure linear ($\delta = 0$) and pure additive ($\gamma = 0$), \model~has no advantage over the lasso or \spam.

When $N_{train} = 200$, both \spam~and \model~overfit significantly
when there are a lot of nonlinear components, since a large number of
nonlinear components leads to a large parameter space and this small
amount of data is not enough for reliable estimates. The lasso wins
over the other methods on most of the cases by trading off variance
for bias. \model~outperforms the lasso in regimes with a mixture of
small nonlinear components and a reasonable number of linear
components. When we increase the number of data points in the training
set ($N_{train} = 500$), more reliable estimates can be obtained so
\spam~wins back from the lasso on cases where we only have nonlinear
components (the lasso, having only linear features is incapable of
estimating the nonlinear effects present in the data). Interestingly,
the lasso is still the best when there are a lot of nonlinear
components since in this regime the data cannot support the large
number of parameters for reliable estimation. \model, however, is the
winner in most settings since it can better model the mixture of linear and nonlinear effects when there are enough data. Not surprisingly, when there are enough data ($N_{train} = 1000$), \model~dominates all cases in which both linear and nonlinear components are present. This is because the lasso is unable to model nonlinear effects and because \spam~has higher variance than \model~without being less biased. 

\begin{table}[t]
    \centering
    \caption{\em \small Size (total number of points) and dimension of datasets and performance
      of methods. For each method, we report the mean error (standard
      deviation in parentheses) and how many of the selected features
      are nonlinear, written as $|\hat{\mathcal N}|/(|\hat{\mathcal N}|+|\hat{\mathcal L}|)$. Bold indicates the method with the mean lowest error.}
    \label{tbl:real}
    {\footnotesize
    \begin{tabular}{|c|c|c|c||c|c|c|c|c|c|}
    \hline
    {\bf Dataset} & {\bf Size} & {\bf Test} & $\mathbf p$ &
    \multicolumn{2}{|c|}{{\bf Lasso}} & \multicolumn{2}{|c|}{{\bf \model}}
    & \multicolumn{2}{|c|}{{\bf SpAM}}\\
    \hline
    Spambase & 4601 & 920 & 57 & 7.38 (0.87) & 0/52 & \textBF{6.57 (0.91)} & 38/41 & 6.93 (0.96) & 38/38\\
    Gisette & 6000 & 1200 & 5000 & 2.43 (0.54) & 0/717 & \textBF{2.18 (0.59)} & 10/733 & 2.62 (0.51) & 1364/1461\\
    RCV1 & 697641 & 418584 & 47236 & 2.71 (0.02) & 0/7652 & \textBF{2.67 (0.01)}  & 4/5293 & 3.18 (0.04) & 4498/4683\\
    Pantheon & 62849 & 37709 & 10000 & 9.34 (0.12) & 0/1859 & \textBF{9.22 (0.16)} & 27/1853 & 12.71 (0.19) & 2770/2770\\
    \hline
    \end{tabular}
    }
\end{table}

\subsection{Real Problems}
\label{sec:real}
In this section, we report experimental results on several real classification problems. We choose datasets with different dimensions and sizes. Table~\ref{tbl:real} summarizes the characteristics of the datasets and presents the predictive performance of the lasso, \model, and \spam~with means and standard deviations on 5 trials. Best parameters are chosen on a held-out validation set on each trial. We note that \model~outperforms the lasso and \spam~on most of the trials. We also list the number of selected nonlinear features and total number of selected features in Table~\ref{tbl:real}. In our experiments, features are forced to be linear if they have less than 10 unique values. This is the case with \spam~on Gisette and RCV1 dataset.

\textbf{Email Classification}. We first consider a classification problem for detecting spam emails (Spambase)~\citep{hastie2001elements}. The features include statistics of particular words or letters in an email. We see from Table~\ref{tbl:real} that by allowing features to act nonlinearly, the error of \spam~decreases substantially compared to the lasso. However, by explicitly setting some of the variables to stay linear, \model~further outperforms \spam.

\textbf{Handwritten Digit Recognition}. We use the ``Gisette" dataset
constructed from NIPS 2003 feature selection
challenge (\url{http://www.nipsfsc.ecs.soton.ac.uk/}). The
problem is to separate the highly confusible digits ``4" and
``9". Features in this dataset contain pixels that are necessary to
distinguish ``4" from ``9", but higher order features from those
pixels as well as random noise features are also added. Since the dimension of this dataset is significantly larger than the previous dataset while the size of the dataset remains similar, we expect \spam~to overfit as shown in Table~\ref{tbl:real}. In our experiments, the best \spam~model that we can get is always worse than the lasso on each cross validation set while our \model~outperforms the lasso on most cross validation sets. Our \model~selects about 733 features, with about 10 of them being nonlinear and the rest being linear, while the lasso selects 717 features. This confirms that by allowing a small number of features to act nonlinearly, we can further improve the classification performance, and yet by setting most of features as linear, we effectively control the complexity and avoid overfitting.

\begin{figure*}[t]
\begin{center}
\begin{tabular}{ccc}
    \includegraphics[width=0.26\textwidth]{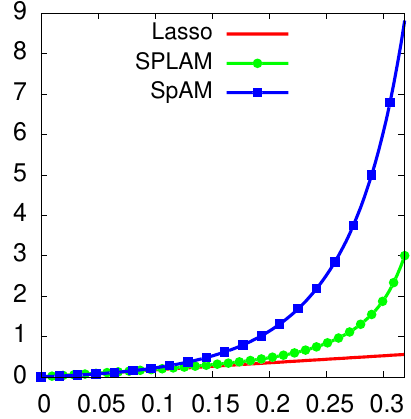} & \includegraphics[width=0.26\textwidth]{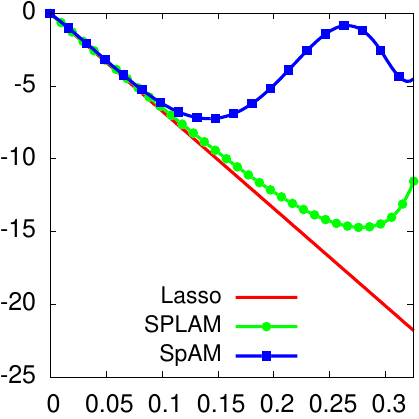} & \includegraphics[width=0.26\textwidth]{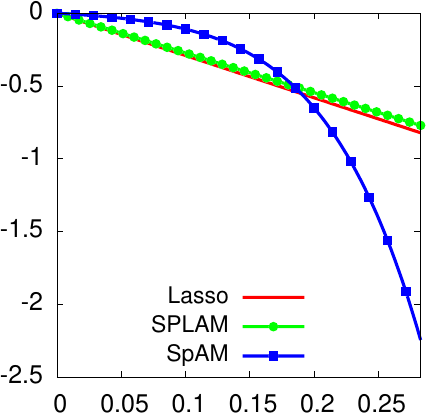}\\
    (a)  & (b) & (c)
\end{tabular}
\end{center}
\caption{\em\small Estimated component functions on RCV1 dataset.} 
\label{fig:rcv1}
\end{figure*}

\textbf{Text Categorization}. Text categorization is an important task
for many natural language processing applications. We use Reuters
Corpus Volume I (RCV1) which involves binary
classification~\citep{lewis2004rcv1}. From Table~\ref{tbl:real}  we see that
\model~outperforms the others. This suggests that in high dimensions, there is extra accuracy that can be
obtained over the lasso if some features are allowed to be nonlinearly
transformed. However, if all features are allowed to be nonlinearly
transformed, such as in \spam, the model will overfit and a suboptimal
model is obtained. On average \model~selects 5293 features with 4 of
them being nonlinear, and the lasso selects 7652
features. Figure~\ref{fig:rcv1} visualizes some
components in \model, \spam, and the lasso. Each feature in this
dataset relates to the (normalized) frequency of some word in a
document. In general, we find that \model~strikes a compromise between
the lasso and \spam~fits, using nonlinearity only sparingly. Figure~\ref{fig:rcv1} (a) shows a feature which is
identified as nonlinear by both \model~and \spam. Notice that both \model~and
\spam~find a model with similar shape. In Figure~\ref{fig:rcv1} (b),
we show a feature that appears to be nearly linear in \model. In this case,
\spam~oscillates in a way that might suggest that it is overfitting to
the noise; by contrast, \model's fit is mostly linear, only exhibiting nonlinear effects when the
input gets large. Notice \model~and the lasso almost agree with each
other when the feature value is small. Finally, in Figure~\ref{fig:rcv1}
(c) we show a feature that is
linear in both \model~and the lasso. \model's estimation of the slope is very similar to that of the
lasso.  By contrast, \spam~treats this as a nonlinear effect. In light of \model's better misclassification rate in this
data set, one might suppose that \spam's more pronounced deviations
from linearity are in fact cases of overfitting to noise.  Likewise,
\model's better misclassification rate compared to the lasso
suggests that the latter may be failing to model some of the nonlinear effects.  

\textbf{Image Matching}. Many new computer vision applications are
utilizing large-scale datasets of places derived from the many
billions of photos on the Web. Image matching is a central procedure
to those applications which tests whether two images are geometrically
consistent~\citep{lou2012matchminer}. Since image matching is an
expensive procedure, image pairs are usually pre-filtered with a
lightweight classification procedure to estimate whether two images
are likely to pass the geometric verification. In this study, we use
the ``Pantheon" dataset in~\citet{lou2012matchminer}. Each image is
represented using bag-of-visual-words model with a vocabulary of
10,000 visual words. From Table~\ref{tbl:real} we again observe that
by carefully controlling the complexity of the model, \model~has
better predictive performance than the other two models. On average
the lasso selects 1859 features while
\model~selects 1853 features with only 27 of them being nonlinear.  

\section{Conclusion}
\label{sec:conclusion}

In this paper, we introduce the sparse partially linear additive model that performs two model-selection tasks within a single convex hierarchical sparse regularization problem.  This formulation permits an efficient optimization algorithm, making the GPLAM framework practical in machine learning settings. We develop an oracle inequality of \model~that makes no assumptions on the design matrix, and we study \model's advantage over \spam~when many of the features in the model are linear. Our thorough experiments demonstrate that \model~can effectively and accurately find relevant components with proper complexity and is very competitive for additive modeling. In particular, on large-scale, high-dimensional datasets, \model~improves accuracy over the popular linear model by allowing a small set of features to have a nonlinear effect.

\section{Acknowledgments}
\label{sec:acknowledgments}

The authors gratefully acknowledge Johannes Lederer for useful discussions. This research has been supported by the NSF under Grants IIS-0911036 and IIS-1012593.  

\appendix
\section{Convergence of Algorithm~\ref{algo:bcg} in the Main Paper}
\label{sec:convergence}

We show that Algorithm~\ref{algo:bcg} fits the
general BCGD framework~\citep{tseng2009coordinate} and therefore the
global convergence is guaranteed.  We include this supplementary material for
completeness although a similar convergence result for the group lasso
is shown in~\citet{qin2010efficient}. We first briefly review the general BCGD algorithm.

Let $F(\beta) = L(\beta) + h(\beta)$, where $h(\beta) = \lambda\Omega^{\model}(\beta)$. At each iteration $k$, for block $j$, choose a symmetric positive definite matrix $H^k$, and compute the search direction.
\begin{align}\label{eqn:bcgd_step}
d^k = \arg\min_{d}\{\nabla L(\beta^k)^Td + \frac{1}{2}d^TH^kd + h(\beta^k + d)\}
\end{align}
where $\forall i \not\in \mathcal{G}_j, d_i = 0$. Then a step size $\alpha^k>0$ is chosen so that the following Armijo rule is satisfied,
\begin{align}\label{inequality:armijo}
F(\beta^k + \alpha^kd^k) \le F(\beta^k) + \alpha^k\sigma\Delta^k
\end{align}
where $0 < \sigma < 1, 0 \le \gamma < 1$, and
\begin{align}
\Delta^k \overset{\text{def}}{=} \nabla L(\beta^k)^Td^k + \gamma {d^k}^TH^kd^d + h(\beta^k + d^k) - h(\beta^k),
\end{align}
Once the step size $\alpha^k$ is determined, update $\beta^{k + 1} = \beta^k + \alpha^kd^k$.

Theorem 2 in~\cite{tseng2009coordinate} guarantees the global convergence when $\overline{\theta}I \succeq H^k \succeq \underline{\theta}I$, $0 < \underline{\theta} \le \overline{\theta}$.

\begin{theorem}
Algorithm~\ref{algo:bcg} fits the general BCGD framework of~\cite{tseng2009coordinate}. The global convergence is guaranteed and Algorithm~\ref{algo:bcg} converges Q-linearly.
\end{theorem}
\begin{proof}
First, for block $j$, setting $H^k = \frac{1}{t_j}I$,
Equation~\eqref{eqn:bcgd_step} is equivalent to our proximal operator
for block $j$ after ignoring constants. Next, notice that when
$\alpha^k = 1, \sigma = 1$, and $\gamma = \frac{1}{2}$, the Armijo
rule becomes our backtracking line search step in Equation~\eqref{eqn:backtracking} in the main paper. That is, the effort of choosing step size is shifted to finding $H^k$. Besides, Lemma 1 in~\cite{tseng2009coordinate} suggests $\nabla L(\beta^k)^Td^k + {d^k}^TH^kd^k + h(\beta^k + d^k) - h(\beta^k) \le 0$. Since $H^k \succ 0$, with $\gamma = \frac{1}{2}$, we can easily see $\Delta^k \le 0$ whenever $d^k \neq 0$, which means if the Armijo rule holds for $\sigma = 1$, it must also hold for $\sigma < 1$. Finally, we show that $\overline{\theta}I \succeq H^k \succeq \underline{\theta}I$. Assume the initial step size is $t_j^0$, this is true when $\overline{\theta} = \max\{C_j, 1/t_j^0\}$ and $\underline{\theta} = \min\{C_j, 1/t_j^0\}$. Thus, according to Theorem 2 in~\cite{tseng2009coordinate}, Algorithm 1 converges Q-linearly.
\end{proof}

\section{Practical Issues}
\label{sec:practical_issues}

\subsection{Active Set Strategy}
\label{sec:active_set}
We employ the widely used active set strategy~\citep{friedman2010regularization, krishnapuram2005sparse, meier2008group}. After a complete cycle through all the variables, we iterate only on the active set till convergence. If another complete cycle does not change the active set, we are done, otherwise the process is repeated.

\subsection{Regularization Path}
\label{sec:parameter}
Similar to \texttt{glmnet}~\citep{friedman2010regularization}, the
optimization of \model~also uses two parameters, $\lambda$ and
$\alpha$, which usually involves a grid search on values of $(\lambda,
\alpha)$ pairs. As noted in the main paper, for each value of
$\alpha$, we start at the smallest value $\lambda_{max}$ for which
$\beta_j = 0$ for $j = 1, ..., p$.  We then decrease $\lambda$ from
$\lambda_{max}$ exponentially.  To find $\lambda_{max}$, we note that
for all $\lambda \ge \lambda_{init} = \max_j\frac{\|\nabla_{\beta_j}
  L(0)\|_2}{\alpha}$, the zero vector is the solution to our optimization problem. We perform a binary search to find $\lambda_{max}$. As described in Algorithm~\ref{algo:lambda}, we start with $\lambda_{init}$ (Line 1) and effectively shrink the interval $[\lambda_l, \lambda_h]$ (Line 3 - 8) to locate $\lambda_{max}$.

\begin{algorithm}[t]
\caption{Finding $\lambda_{max}$}
\label{algo:lambda}
\begin{algorithmic}[1]
\STATE $\lambda_h \leftarrow \max_j\frac{\|\nabla_j L(0)\|_2}{\alpha}$
\STATE $\lambda_l \leftarrow 0$
\WHILE{$\lambda_h - \lambda_l \ge \epsilon$}
    \STATE $\lambda \leftarrow \frac{\lambda_h + \lambda_l}{2}$
    \IF{$\forall j, P_t^j(0) = 0$}
        \STATE $\lambda_h \leftarrow \lambda$
    \ELSE
        \STATE $\lambda_l \leftarrow \lambda$
    \ENDIF
\ENDWHILE
\STATE $\lambda_{max} = \lambda_h$
\end{algorithmic}
\end{algorithm}

\section{Proof of Theorem~\ref{thm:oracle_slow_rate} and Corollary~\ref{thm:slow_rate} in the Main Paper}
\label{appendix:thm1_proof}

\begin{proof}[Proof of Theorem~\ref{thm:oracle_slow_rate}]
By definition of $\hat{\beta}$,
\begin{align*}
\frac1{2N}\|y-X\hat\beta\|_2^2+\lambda\Omega^{\model}(\hat\beta)\le \frac1{2N}\|y-X\beta\|_2^2+\lambda\Omega^{\model}(\beta)
\end{align*}
holds for any $\beta\in\mathbb R^p$.  Some algebra (recalling that $y=f^0+\epsilon$ and writing $\hat\Delta=\hat\beta-\beta$) leads to
\begin{align}
\frac1{2N}\|X\hat\beta-f^0\|_2^2 + \lambda\Omega^{\model}(\hat\beta)\le \frac1{2N}\|X\beta-f^0\|_2^2+\frac1{N}\epsilon^TX\hat\Delta+\lambda\Omega^{\model}(\beta)\label{splam:eq:bound}
\end{align}
Define the empirical process as,
\begin{align}
V_N(\hat\Delta) = \frac1{N}\epsilon^TX\hat\Delta = \frac1{\sqrt{N}}\sum_{j = 1}^p V_j^T\hat\Delta_j
\end{align}
where $V_j = \frac{1}{\sqrt{N}}X_j^T\epsilon \in \mathbb{R}^M$.

Now we bound the empirical process. First we notice that,
\begin{align}
|V_j^T\hat\Delta_j| \le& \frac1{2} \left[ |V_j^T\hat\Delta_j| + |V_{j1}\hat\Delta_{j1}| + |V_{j, -1}^T\hat\Delta_{j, -1}| \right]\\
\le& \frac1{2}\left[ \|V_j\|_2\|\hat\Delta_j\|_2 + |V_{j1}||\hat\Delta_{j1}| + \|V_{j, -1}\|_2\|\hat\Delta_{j, -1}\|_2\right]
\end{align}

Thus $|V_N(\hat\Delta)|$ can be bounded as follows.
\begin{align}
|V_N(\hat\Delta)| \le& \frac1{\sqrt{N}}\sum_{j = 1}^p |V_j^T\hat\Delta_j|\\
\le& \frac{1}{2\sqrt{N}}\left(\max_j \|V_j\|_2\|\hat\Delta\|_{2,1} + \max_j |V_{j1}|\sum_j|\hat\Delta_{j1}| + \max_j\|V_{j, -1}\|_2\|\hat\Delta_{\cdot, -1}\|_{2,1} \right)\\
\le& \frac{1}{2\sqrt{N}}\left[(\max_j\|V_j\|_2 + \max_j|V_{j1}|)\|\hat\Delta\|_{2, 1} + \max_j\|V_{j, -1}\|_{2}\|\hat\Delta_{\cdot, -1}\|_{2,1}\right]
\end{align}

Observing that $V_j \sim N(0, \sigma^2I_M)$, we have $\|V_j\|_2^2 \sim \sigma^2\chi^2_M$. Thus, by Lemma 6.2 and 8.1 of~\cite{buhlmann2011statistics}, we have
\begin{align}
P\left( \frac{\max_j |V_{j1}|}{2\sqrt{N}} > \nu_1 \right) \le& 2e^{-x}\\
P\left( \frac{\max_j \|V_{j}\|_2}{2\sqrt{N}} > \nu_2 \right) \le& e^{-x}\\
\end{align}
where,
\begin{align}
\nu_1^2 =& \frac{\sigma^2}{2N}(x + \log p)\\
\nu_2^2 =& \frac{\sigma^2}{4N}\left[M + \sqrt{4M(x + \log p)} + 4(x + \log p)\right]
\end{align}

Thus, we have
\begin{align}
P\left(\frac{\max_j\|V_j\|_2 + \max_j|V_{j1}|}{2\sqrt{N}} > \nu_1 + \nu_2\right) \le& 3e^{-x}\\
P\left(\frac{\max_j\|V_{j, -1}\|_2}{2\sqrt{N}} > \nu_2\right) \le& e^{-x}
\end{align}

Therefore (with union bound),
\begin{align}
P\left(|V_N(\hat\Delta)| \le \left[ (\nu_1 + \nu_2)\|\hat\Delta\|_{2, 1} + \nu_2\|\hat\Delta_{\cdot, -1}\|_{2,1}\right]\right) \ge& 1 - (e^{-x} + 3e^{-x})\\
=& 1 - 4e^{-x}
\end{align}

Thus, by~\eqref{splam:eq:bound} we have with probability at least $1 - 4e^{-x}$ that
\begin{align}
\frac1{2N}\|X\hat\beta-f^0\|_2^2 + \lambda\Omega^{\model}(\hat\beta)\le \frac1{2N}\|X\beta-f^0\|_2^2+(\nu_1 + \nu_2)\|\hat\Delta\|_{2,1} + \nu_2\|\hat\Delta_{\cdot, -1}\|_{2,1}+\lambda\Omega^{\model}(\beta)
\label{splam:eq:common_bound}
\end{align}

Let $\lambda_1 = \lambda\alpha$ and $\lambda_2 = \lambda(1 - \alpha)$, we can take $\lambda_1 = 2(\nu_1 + \nu_2)$ and $\lambda_2 = 2\nu_2$. Thus, \eqref{splam:eq:common_bound} implies
\begin{align}
\frac1{2N}\|X\hat\beta-f^0\|_2^2 &-\frac1{2N}\|X\beta-f^0\|_2^2\le 
(\lambda/2)\Omega^{\model}(\hat\Delta) -
\lambda\Omega^{\model}(\hat\beta) + \lambda\Omega^{\model}(\beta)\\
\le& (\lambda/2)\left[\Omega^{\model}(\hat\beta)
+\Omega^{\model}(\beta)\right] -
\lambda\Omega^{\model}(\hat\beta) + \lambda\Omega^{\model}(\beta)\\
=& (3\lambda/2)\Omega^{\model}(\beta)-(\lambda/2)\Omega^{\model}(\hat\beta)
\end{align}
by the triangle inequality.  Thus,
\begin{align}
\frac1{2N}\|X\hat\beta-f^0\|_2^2 \le \frac1{2N}\|X\beta-f^0\|_2^2 +3\lambda\Omega^{\model}(\beta).
\end{align}

By choosing $x = \log p$, we can ensure our inequality holds with probability at least $1 - 4/p$. This means,
\begin{align}
\nu_1^2 =& \frac{\sigma^2}{N}\log p\\
\nu_2^2 =& \frac{\sigma^2}{4N}\left[ M + \sqrt{8M\log p} + 8\log p\right].
\end{align}

Define $\tilde\nu_1^2\overset{\text{def}}{=}\frac{\sigma^2}{N}\log p$ and notice that $\nu_2^2 \le 6\sigma^2\log p / N\overset{\text{def}}{=}\tilde\nu_2^2$ if $\log p \ge M/8$. Now, as long as $\log p \ge M/8$, we can take $\lambda \ge 2(\tilde\nu_1 + 2\tilde\nu_2) = 2(1 + 2\sqrt{6})\sigma\sqrt{\log p / N}$ and
\begin{align}
\alpha = \frac{\tilde\nu_1 + \tilde\nu_2}{\tilde\nu_1 + 2\tilde\nu_2} = \frac{1 + \sqrt{6}}{1 + 2\sqrt{6}},
\end{align}
with probability at least $1 - 4/p$, we have
\begin{align}
\frac1{2N}\|X\hat\beta-f^0\|_2^2 \le \frac1{2N}\|X\beta-f^0\|_2^2 +3\lambda\Omega^{\model}(\beta).
\end{align}
This holds simultaneously for all $\beta$; this may be succinctly
expressed by adding $\min_\beta$ to the right hand side.

\end{proof}

\begin{proof}[Proof of Corollary~\ref{thm:slow_rate}]
We plug $\beta^0$ into the right-hand side of Theorem~\ref{thm:oracle_slow_rate} and observe that
\begin{align}
  \Omega^{\model}(\beta^0)&=\alpha\sum_{j\in
    \mathcal{S}^0}\|\beta^0_j\|_2+(1-\alpha)\sum_{j\in \mathcal{N}^0}\|\beta_{j,-1}^0\|_2\\
  &\le \alpha\sum_{j\in \mathcal{L}^0}\|\beta^0_j\|_2+\sum_{j\in
    \mathcal{N}^0}\|\beta_{j}^0\|_2\\
  &=\alpha\sum_{j\in \mathcal{L}^0}|\beta^0_{j1}|+\sum_{j\in \mathcal{N}^0}\|\beta_{j}^0\|_2
\end{align}

\end{proof}

\section{Proof of Lower Bound on \spam's Prediction Error}
\label{sec:proof-lower-bound}

We assume that all $p$ features are linear with equal coefficients,
i.e., $\beta_{j}^0=be_1\in\mathbb R^M$
and consider an asymptotic regime in which $p$ is fixed and $N=pM$, with $M,N\to\infty$.  We assume that all
features are orthogonal, i.e., $\frac1{N}X^TX=I_{pM}$.  In the main
paper, we note that \spam~in this case is given by the expression:
$$
\hat\beta_j^{\spam}=\gamma_j(\lambda)\frac1{N}X_j^Ty\quad\text{where}\quad\gamma_j(\lambda)=\left(1-\frac{\lambda}{\|\frac1{N}X_j^Ty\|_2}\right)_+.
$$
Now, $\frac{1}{N}X_j^Ty=be_1+U_j$ where
$U_j=\frac{1}{N}X_j^T\epsilon\sim N(0,\frac{\sigma^2}{N}I_M)$.  Since
$\|\frac1{N}X_j^Ty\|_2^2\to b^2+\sigma^2/p$, asymptotically,
the shrinkage factor $\gamma_j(\lambda)=\gamma$ is a nonrandom value,
not depending on $j$, and the prediction error is
\begin{align*}
  \frac1{N}\|X\hat\beta^{\spam}-X\beta^0\|^2&=\sum_{j=1}^p\|\gamma(be_1+U_j)-be_1\|^2\\
&=\gamma^2(b^2p+\sum_{j=1}^p\left[\|U_j\|^2
    +2bU_{j1}\right])+pb^2-2\gamma\sum_{j=1}^pb(b+U_{j1})\\
&\to \gamma^2(b^2p+\sigma^2)+pb^2-2\gamma pb^2.
\end{align*}
For the best possible asymptotic error, we can choose
$\gamma=pb^2/(pb^2+\sigma^2)$ (equivalent to choosing the best
$\lambda$).  At this value, 
$$
\lim_{N\to\infty}\frac1{N}\|X\hat\beta^{\spam}-X\beta^0\|^2\ge \frac{b^2}{1/b^2+p/\sigma^2}>0.
$$
Thus, \spam~is not consistent in terms of prediction error in this asymptotic regime.

To see that \model~with $\lambda\alpha=0$ and
$\lambda(1-\alpha)=\infty$ is consistent in terms of prediction error,
observe that  
$\hat\beta^{\model}=(X_{j1}^Ty)e_1=(b+U_{j1})e_1$ and
$$
\frac1{N}\|X\hat\beta^{\model}-X\beta^0\|^2=\sum_{j=1}^p\|(b+U_{j1})e_1-be_1\|^2=\sum_{j=1}^pU_{j1}^2\sim\frac{\sigma^2}{N}\chi^2_p\to0.
$$

\section{Experiments}
\label{appendix:experiments}

\begin{figure*}[t]
\begin{center}
\begin{tabular}{cccc}
    \includegraphics[width=37.5mm]{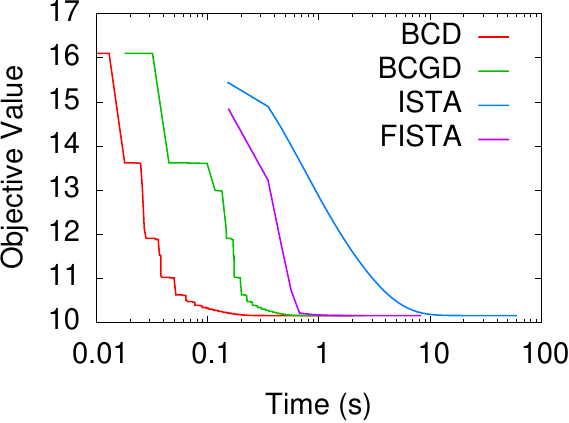} & \includegraphics[width=37.5mm]{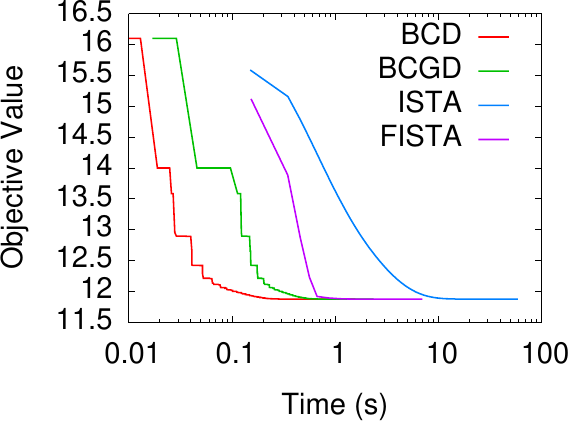} & \includegraphics[width=37.5mm]{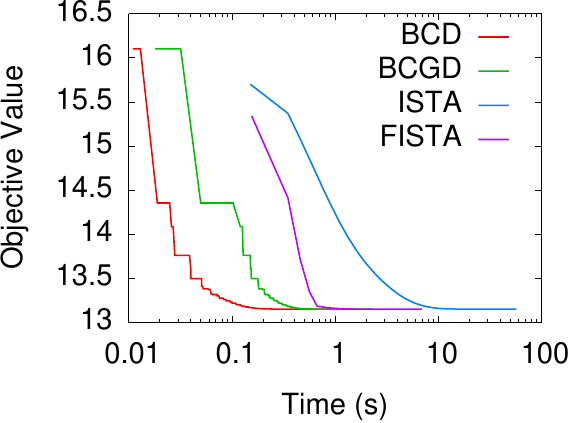} & \includegraphics[width=37.5mm]{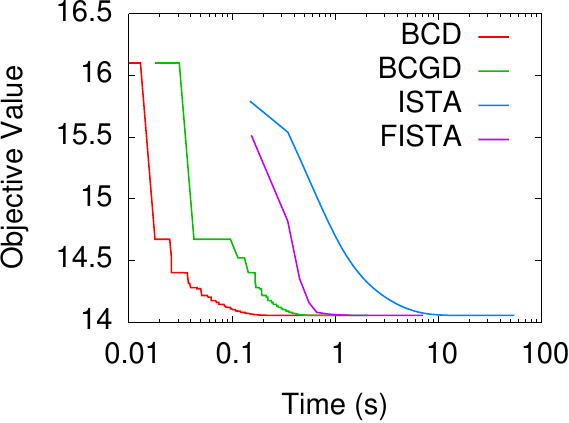}\\
    (a) $\lambda = 0.01, \alpha = 0.1$ & (b) $\lambda = 0.01, \alpha = 0.15$ & (c) $\lambda = 0.01, \alpha = 0.2$ & (d) $\lambda = 0.01, \alpha = 0.25$\\
    \includegraphics[width=37.5mm]{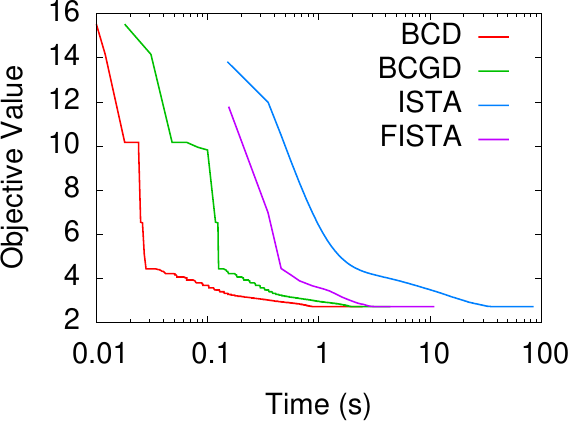} & \includegraphics[width=37.5mm]{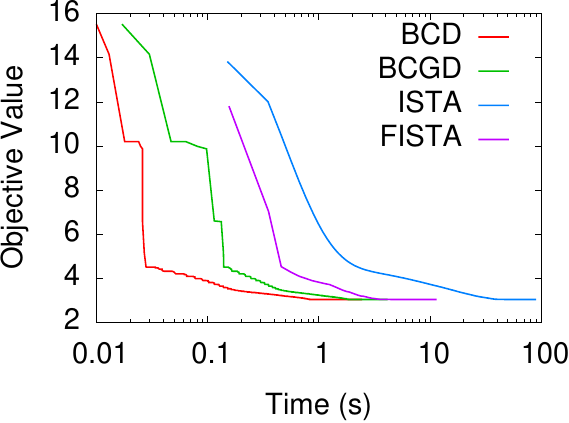} & \includegraphics[width=37.5mm]{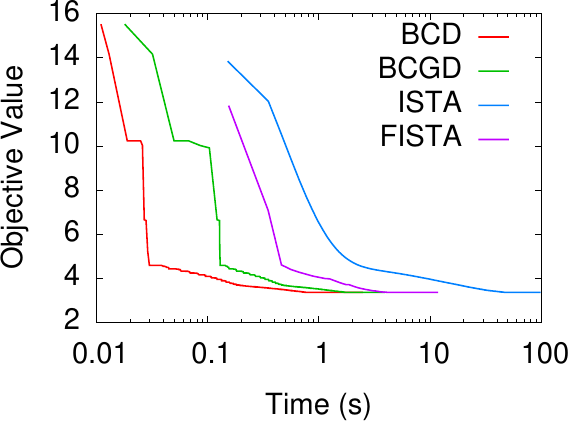} & \includegraphics[width=37.5mm]{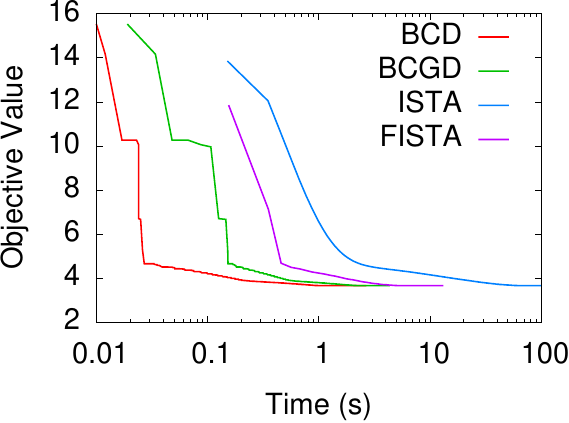}\\
    (e) $\lambda = 0.001, \alpha = 0.1$ & (f) $\lambda = 0.001, \alpha = 0.15$ & (g) $\lambda = 0.001, \alpha = 0.2$ & (h) $\lambda = 0.001, \alpha = 0.25$
\end{tabular}
\end{center}
\caption{\em Objective value vs. running time for synthetic dataset in Section 5.1.}
\label{splam:fig:synth_time}
\end{figure*}

In this section, we compare our BCGD algorithm and BCD algorithm with ISTA and FISTA~\citep{beck2009fista} using the synthetic function in Section~\ref{sec:synthetic}. We report running time of all the methods on a single core. For BCGD, ISTA, and FISTA, we start with a same initial step size. For fair comparison, we turn off the active set strategy in BCGD and BCD, and we directly use the design matrix after QR decomposition so that all methods are applied to the same optimization problem.

Figure~\ref{splam:fig:synth_time} illustrates the running time for all methods using the same synthetic dataset in Section~\ref{sec:synthetic} for different combinations of $\lambda$ and $\alpha$. As expected, FISTA converges much faster than ISTA. However, the BCGD algorithm is faster than both of the these methods. This is because BCGD uses more information in the sense of more frequent updates. In addition, we can see that BCD further speeds up the optimization since there is no step size in BCD; this not only solves exactly the subproblem but also avoids the possibility of dampening the step size and repeating the computation on the same block.

\bibliography{sigproc} 

\end{document}